\documentclass[11pt]{article}

\usepackage{amssymb,amsfonts,amsmath,amsthm}
\usepackage{geometry}
\usepackage{color,cite,enumerate,graphicx}

\newcommand{\eps}{\varepsilon}
\newcommand{\R}{\mathbb R}

\newcommand{\m}{\mathbf{m}}
\newcommand{\mpa}{m_\parallel}

\newcommand{\hpa}{h_\parallel}
\newcommand{\hpe}{\mathbf h_\perp}
\newcommand{\mpe}{\mathbf{m}_\perp}
\newcommand{\wmpa}{\widetilde{m}_\parallel}
\newcommand{\wmpe}{\widetilde{\mathbf{m}}_\perp}
\newcommand{\vv}{{\mathbf v}}

\newtheorem{theorem}{Theorem}

\newtheorem{corollary}[theorem]{Corollary}
\newtheorem{remark}[theorem]{Remark}

\numberwithin{equation}{section}

\title{Domain structure of ultrathin ferromagnetic elements in the
  presence of Dzyaloshinskii-Moriya interaction}

\author{Cyrill B. Muratov\footnote{Department of Mathematical
    Sciences, New Jersey Institute of Technology, Newark, NJ 07102,
    USA} \and Valeriy V. Slastikov\footnote{School of Mathematics,
    University of Bristol, Bristol BS8 1TW, United Kingdom}}

\date{\today}

\begin{document}

\maketitle

\begin{abstract}
  Recent advances in nanofabrication make it possible to produce
  multilayer nanostructures composed of ultrathin film materials with
  thickness down to a few monolayers of atoms and lateral extent of
  several tens of nanometers. At these scales, ferromagnetic materials
  begin to exhibit unusual properties, such as perpendicular
  magnetocrystalline anisotropy and antisymmetric exchange, also
  referred to as Dzyaloshinskii-Moriya interaction (DMI), because of
  the increased importance of interfacial effects. The presence of
  surface DMI has been demonstrated to fundamentally alter the
  structure of domain walls. Here we use the micromagnetic modeling
  framework to analyse the existence and structure of chiral domain
  walls, viewed as minimizers of a suitable micromagnetic energy
  functional. We explicitly construct the minimizers in the
  one-dimensional setting, both for the interior and edge walls, for a
  broad range of parameters. We then use the methods of
  $\Gamma$-convergence to analyze the asymptotics of the
  two-dimensional magnetization patterns in samples of large spatial
  extent in the presence of weak applied magnetic fields.
\end{abstract}

\section{Introduction}
\label{sec:introduction}

The exploding amount of today's digital data calls for revolutionary
new high density, fast and long-term information storage
solutions. Spintronics is one among the emerging fields of
nanotechnology offering a great promise for information technologies,
whereby information is carried and processed, using the electron spin
rather than its electric charge \cite{prinz98, zutic04, allwood05,
  bader10}. It brings about many opportunities for creating the next
generation of devices combining spin-dependent effects with
conventional charge-based electronics. Despite being a relatively new
field of applied physics, it has already firmly established its
presence in everyday life through the development of new magnetic
storage devices. The discovery of giant magnetoresistance (GMR), for
which A. Fert and P. Gr\"unberg were awarded the 2007 Nobel Prize in
Physics, allowed an ability to ``read'' the magnetization states of a
ferromagnet through electric resistance measurements. This effect has
been used in GMR-based spin valves, which transformed magnetic
hard-disk drive technology, leading to increases in storage density by
several orders of magnitude. Yet, the GMR magnetic storage technology
has already been superseded by novel spin-dependent devices based on
the effect of tunneling magnetoresistance (TMR), another exciting
development in the field of spintronics \cite{bader10}.

Recent discoveries of new physical phenomena that become prominent at
nanoscale open up a possibility of unprecedented data storage
densities and read/write speeds. These include spin transfer torque
(STT), chiral domain walls and magnetic skyrmions, spin Hall effect,
spin Seebeck effect, electric field control of the magnetic
properties, etc. (see, e.g., \cite{bader10, brataas12, fert13,
  nagaosa13, chen13, vonbergmann14, matsukura15}). The ability to
manipulate the magnetization, using electric currents suggests novel
designs for magnetic memory.  One popular concept is the so-called
racetrack memory \cite{bader10, parkin08}, which uses a
two-dimensional array of parallel nanowires where magnetic domains --
``bits'' -- may be read, moved, and written through an application of
a spin current. Another promising type of memory and logic devices is
based on storing and manipulating the data bits, using magnetic
skyrmions, rather than magnetic domain walls. The existence of
magnetic skyrmions was predicted theoretically more than twenty-five
years ago \cite{bogdanov89,bogdanov94}, but their experimental
observations are much more recent \cite{nagaosa13, heinze11,
  woo16}. The topological stability, small size and extremely low
currents and fields required to move magnetic skyrmions make them
natural candidates for the use in spintronic memory and logic devices
\cite{fert13, zhang15sr, woo16}.

A successful design of novel spintronic devices that make use of
magnetic domain walls or skyrmions is strongly dependent on a deep
theoretical understanding of static and dynamic behaviors of the
magnetization in magnetic nanostructures. The manipulation and control
of magnetic domain walls and topologically protected states (e.g.,
magnetic vortices and skyrmions) in ferromagnetic nanostructures has
been the subject of extensive experimental and theoretical research
(see, e.g., \cite{braun12, thiaville12, chen13, rohart13, goussev13,
  sampaio13, boulle13}; this list is certainly far from
complete). Recent advances in nanofabrication techniques
\cite{stepanova} have lead to the production of ultrathin films with
thickness down to several atomic layers and a lateral extent down to
tens of nanometers. These ultrathin magnetic films and multilayer
structures often exhibit unusual magnetic properties, attributed to an
increased importance of interfacial effects.  The most important
features of these ultrathin magnetic structures include the appearance
of perpendicular magnetic anisotropy \cite{heinrich93,ikeda10} and the
Dzyaloshinskii-Moriya interaction (DMI) \cite{dzyaloshinskii58,
  moriya60}. The latter is closely related to reflection symmetry
breaking in such films and leads to emergence of magnetization
chirality \cite{fert90, hrabec14, thiaville12}.

\begin{figure}[t]
  \centering
  \includegraphics[width=14cm]{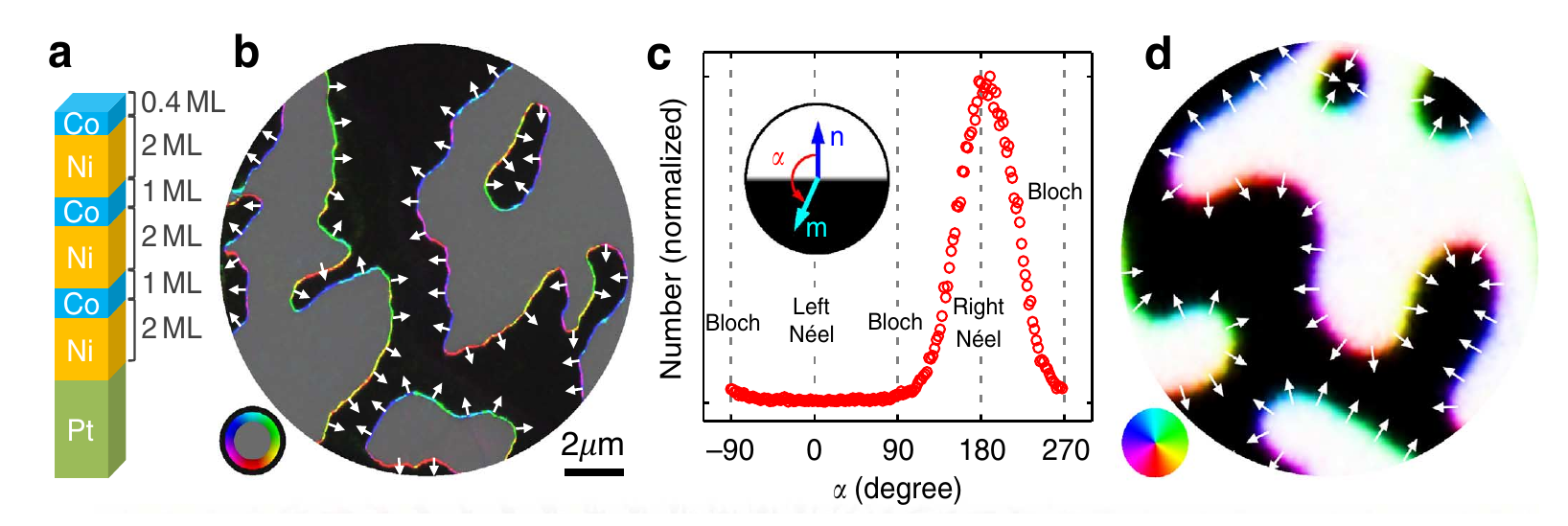}
  \caption{Experimental and numerical observations of chiral domain
    walls in ultrathin ferromagnetic films in the presence of DMI. (a)
    The schematics of the multilayer structure. (b) A colormap of the
    magnetization exhibiting chiral domain walls.  (c) A histogram of
    the in-plane magnetization orientation angle relative to the
    in-plane normal to the domain wall showing a preferred rotation
    direction. (d) A comparison to the result of a Monte-Carlo
    simulation of a discrete spin model. In (b), gray indicates the
    domains with the magnetization up, black indicates the domains
    with the magnetization down, and the rest of the colors correspond
    to the directions of the in-plane component, as shown in the
    color-wheel. Adapted from Ref.~\cite{chen13}, with permission; see
    that reference for further details. }
  \label{fig:expwalls}
\end{figure}

The experimental discovery of the symmetry breaking
Dzyaloshinskii-Moriya interaction in ferromagnetic multilayers has
generated a lot of interest in the physics community \cite{bode07,
  heinze11, romming13}. There has been a lot of work focusing on the
influence of DMI on magnetization configurations within a
ferromagnetic sample \cite{thiaville12, rohart13, bode07}. One of the
interesting features of DMI is its influence on the profile and the
dynamic properties of domain walls \cite{thiaville12, chen13,
  rohart13, emori13}. In addition, it is well-known that DMI may be
responsible for formation of magnetic skyrmions -- topologically
protected states with a quantized topological degree observed in
ultrathin films \cite{nagaosa13, melcher14}.  DMI also plays a crucial
role in defining the orientation of the domain walls and chiral
behavior of the magnetization inside the wall, leading to the
formation of a new type of {\it chiral domain walls}, also referred to
as the Dzyaloshinskii walls \cite{thiaville12}, having rather
different properties than the conventional Bloch and Neel walls
\cite{hubert}. For an illustration of chiral domain walls observed
experimentally and numerically, see Fig. \ref{fig:expwalls}.  In a
recent theoretical work \cite{rohart13}, it was reported that the
interplay between DMI and the boundary of an ultrathin ferromagnetic
sample is responsible for creating another type of domain walls --
{\it chiral edge domain walls}. These walls play a crucial role in
producing new types of magnetization patterns inside a
ferromagnet. For instance, in the presence of a transverse applied
field, chiral edge domain walls provide a mechanism for {\it tilting}
of an interior domain wall in a ferromagnetic strip
\cite{boulle13,mst16}. Moreover, they also significantly modify the
dynamic behavior of the interior domain wall under the action of
current and an applied field \cite{thiaville12}.

In this paper, we study chiral domain walls in ultrathin ferromagnetic
films, using rigorous analytical methods within the variational
framework of micromagnetics. Our goal is to understand the formation
of chiral interior domain walls and chiral edge domain walls, viewed
as local or global energy minimizing configurations of the
magnetization, in samples with perpendicular magnetocrystalline
anisotropy in the presence of surface DMI and weak applied magnetic
fields.  The multi-scale nature of the micromagnetic energy allows for
a variety of distinct regimes characterized by different relations
between the material and geometric parameters, and makes its
investigation a very challenging mathematical problem. Many of these
regimes have been investigated analytically, using modern techniques
of calculus of variations in the context of various ferromagnetics
nanostructures (see, e.g.,\cite{desimone06r}).

Our starting point is a reduced two-dimensional micromagnetic energy,
in which the stray field contributes only a local shape anisotropy
term to the leading order (see \eqref{E} below). This energy gives
rise to a non-convex vectorial variational problem, with a nontrivial
interplay between the boundary and the interior of the domain due to
the DMI term.  We seek to understand the formation and structure of
the domain walls -- transition layers between constant magnetization
states -- that correspond to minimizers of the micromagnetic
energy. The framework for this analysis is provided by the variational
methods of the gradient theory of phase transitions
\cite{modica87}. These types of problems have been extensively studied
in the mathematical community in both scalar \cite{modica87,
  modica87aihp, kohn89, owen90} and vectorial \cite{fonseca89,
  sternberg91} settings. The nontrivial influence of the boundary
within the gradient theory of phase transitions was investigated in
\cite{modica87aihp, owen90}.

We begin by investigating the one-dimensional problems on the infinite
and semi-infinite domains. Here we provide a complete analytical
solution for the global energy minimizers of these one-dimensional
problems, see Theorem~\ref{t:1d} and Theorem~\ref{t:edge},
respectively. Our main tool is a careful analysis of the case of
equality in the vectorial Modica-Mortola type lower bound for the
energy of one-dimensional magnetization configurations. Our analysis
yields explicit profiles for one-dimensional chiral interior and edge
domain walls. These optimal profiles are used later on in the
constructions for the full two-dimensional problem. Our
one-dimensional results confirm the physical intuition of
\cite{rohart13} for a slightly reduced range of the DMI constants.

We then investigate the full two-dimensional energy in the regime of
large domains and small applied fields, using methods of
$\Gamma$-convergence. After a rescaling, this amounts to a study of
the asymptotic behavior of the energy $E_\eps(\m)$ in \eqref{Eeps} as
$\eps \to 0$.  We note that our original problem is vectorial,
constrained ($|\m(x)|=1$), and the energy contains linear gradient
terms in the interior, as well as boundary terms (after integration by
parts), both coming from DMI. Even though the original problem is
vectorial -- and these are notoriously difficult phase transition
problems -- we show that one can reduce our problem to a scalar
setting by decoupling the behavior of the normal magnetization
component $\mpa$, preferring to be equal to $\pm1$, and the in-plane
component $\mpe$, preferring to be $0$, outside the transition layer
and proving that the optimal configuration of $\mpe$ is a function of
$\mpa$ and the layer orientation. This nontrivial observation
significantly simplifies the analysis of the problem and allows us to
use the methods developed in \cite{modica87aihp, owen90} to obtain the
$\Gamma$-limit of the family of micromagnetic energies. The rest of
the proof follows the pattern of the gradient theory of phase
transitions \cite{modica87}, with some modifications to account for
the vectorial and constrained nature of the problem.

With the above tools, we obtain the $\Gamma$-limit, given by
\eqref{E0}, of the family of energies in \eqref{Eeps} with respect to
the $L^1$ convergence of $\mpa^\eps$. The limit energy is geometric,
and its minimizers determine the locations of the chiral domain walls,
which are now curves separating the regions in which $\mpa^0$ changes
sign. As a consequence, we also obtain an asymptotic characterization
of the energy minimizers of $E_\eps$ as $\eps \to 0$. Our main result,
stated in Theorem~\ref{t:Gamma}, indicates that the presence of DMI
significantly modifies the magnetization behavior in ultrathin
magnetic films by creating both interior and edge chiral domain walls.

The paper is organized as follows. In
section~\ref{sec:probl-one-dimens}, we present the solution of the
one-dimensional global energy minimization problem for both the
interior and boundary chiral domain walls. Then, in
section~\ref{sec:two-dimens-probl0} we investigate the full
two-dimensional energy \eqref{E} in the regime of large domains and
small applied fields and study the behavior of the family of
micromagnetic energies in \eqref{Eeps} in the limit as $\eps \to 0$.
Finally, in section~\ref{sec:discussion} we summarize our findings and
discuss several additional modeling aspects of our problem, together
with some possible extensions of our analysis.

\section{Model}
\label{sec:model}

We start by considering a ferromagnetic film of thickness $d$
occupying the spatial domain $\Omega \times (0,d) \subset \R^3$, where
$\Omega \subseteq \R^2$ is a two-dimensional domain specifying the
shape of the ferromagnetic element. Within the micromagnetic framework
\cite{hubert}, the magnetization in the sample is described by the
vector $\mathbf M = \mathbf M(x, y, z)$ of constant length
$|\mathbf M| = M_s$, where $M_s$ is referred to as the saturation
magnetization. The micromagnetic energy in the presence of an
out-of-plane uniaxial anisotropy and an interfacial
Dzyaloshinskii-Moriya interaction (DMI) may be written in the SI units
in the form \cite{bogdanov89,bogdanov94,thiaville12}
\begin{align}
  \label{Ephys}
  E(\mathbf M) 
  & = {A \over M_s^2} \int_{\Omega \times (0,d)}
    |\nabla \mathbf M|^2 \, d^3 r +  {K \over M_s^2} \int_{\Omega \times
    (0,d)} |\mathbf M_\perp|^2 d^3 r  - \mu_0 \int_{\Omega \times (0,d)}
    \mathbf M \cdot \mathbf H \, d^3 r \notag \\
  & \quad + \mu_0 \int_{\R^3} \int_{\R^3}
    {\nabla \cdot \mathbf M(\mathbf r) \, \nabla \cdot \mathbf
    M(\mathbf r') \over 8 \pi | \mathbf r - \mathbf r'|} \, d^3 r
    \, d^3 r' + {D d \over M_s^2} \int_\Omega \Big(
    \overline{M}_\parallel \nabla \cdot \overline{\mathbf 
    M}_\perp - \overline{\mathbf M}_\perp \cdot \nabla
    \overline{M}_\parallel \Big) d^2 r.    
\end{align}
Here we wrote $\mathbf M = (\mathbf M_\perp, M_\parallel)$, where we
defined $\mathbf M_\perp \in \R^2$ and $M_\parallel \in \R$ to be the
components of the magnetization vector $\mathbf M$ that are
perpendicular and parallel to the material easy axis (the $z$-axis),
respectively, and introduced $\overline{\mathbf M}$ which is the trace
of $\mathbf M$ on $\Omega \times \{0\}$. In \eqref{Ephys}, $A$ is the
exchange stiffness, $K$ is the magnetocrystalline anisotropy constant,
$\mathbf M$ has been extended by zero outside the sample and
$\nabla \cdot \mathbf M$ is understood distributionally in $\R^3$,
$\mu_0$ is the permeability of vacuum,
$\mathbf H = \mathbf H(x, y, z)$ is the applied magnetic field, and
$D$ is the Dzyaloshinskii-Moriya interaction constant, following the
standard convention to write $D$ in the units of energy per unit
area. In writing the DMI term in this specific form, we took into
account that it arises as a contribution from the interface between
the magnetic layer and a non-magnetic material and should, therefore,
enter as a boundary term in the full three-dimensional theory.

In the above framework, the equilibrium magnetization configurations
in the ferromagnetic sample correspond to either global or local
minimizers of a non-local, non-convex energy functional in
\eqref{Ephys}. This energy includes several terms, in order of
appearance: the exchange term, which prefers constant magnetization
configurations; the magnetocrystalline anisotropy, which favors
out-of-plane magnetization configurations; the Zeeman, or applied field
term, which prefers magnetizations aligned with the external field;
the magnetostatic term, which prefers divergence-free configurations;
and the surface DMI term, which favors chiral symmetry breaking. The
origin of the latter is the antisymmetric exchange mediated by the
spin-orbit coupling in the conduction band of a heavy metal at the
ferromagnet-metal interface \cite{fert80,fert90,crepieux98}.

The variational problem associated with \eqref{Ephys} poses a
significant challenge for analysis. Therefore, in the following we
introduce a simplified version of the energy in \eqref{Ephys} that is
suitable for ultrathin ferromagnetic films of thickness
$d \lesssim \ell_{ex} = \sqrt{2 A / (\mu_0 M_s^2)}$, where $\ell_{ex}$
is the material exchange length. In this case a two-dimensional model
is appropriate in which the stray field energy can be modeled by a
local shape anisotropy term (see, e.g., \cite{gioia97}; for a more
thorough mathematical discussion of the stray field effect in
ultrathin films with perpendicular anisotropy, see
\cite{kmn16}). Measuring the lengths in the units of $\ell_{ex}$ and
the energy in the units of $A d$, we can rewrite the energy associated
with the magnetization configuration
$\mathbf M(x, y, z) = M_s \m(x, y)$, where
$\m : \Omega \to \mathbb S^2$, as
\begin{align}
  \label{E}
  E(\m) = \int_\Omega \Big\{ |\nabla \m|^2 + (Q
  - 1) |\mpe|^2 - 2 \hpa \mpa - 2 \hpe \cdot \mpe \notag \\
  + \kappa \left( \mpa \nabla \cdot \mpe -
  \mpe \cdot \nabla \mpa \right) \Big\} \, d^2 r,
\end{align}
where we defined $\mpe \in \R^2$ and $\mpa \in \R$ to be the
respective components of the unit magnetization vector $\m$ and
introduced the dimensionless quality factor $Q$ and the dimensionless
DMI strength $\kappa$:
\begin{align}
  \label{Qkappa}
  Q = {2 K \over \mu_0 M_s^2}, \qquad \kappa = D \sqrt{2 \over \mu_0
  M_s^2 A},
\end{align}
where $D$ is the DMI constant \cite{thiaville12}. In \eqref{E}, we
also introduced a dimensionless applied magnetic field
$\mathbf h = (\hpe, \hpa) = \mathbf H / M_s$, with $\hpe \in \R^2$ and
$\hpa \in \R$.

We are interested in the regime in which the film favors
magnetizations that are normal to the film plane, i.e., when $Q > 1$.
Also, since the energy is invariant with respect to the transformation
\begin{align}
  \label{kappam}
  \kappa \to -\kappa, \qquad \mpe \to -\mpe,  \qquad \hpe \to -\hpe, 
\end{align}
without loss of generality we can assume $\kappa$ to be positive.

\section{The problem in one dimension}
\label{sec:probl-one-dimens}

We begin by considering an idealized situation in which the
ferromagnetic film occupies either the whole plane or a half-plane,
which leads to two basic types of domain walls considered below (see
Fig. \ref{fig:walls}). These are the magnetization configurations that
vary in one direction only. In the case of the half-plane, the
magnetization is also assumed to vary in the direction normal to the
film edge. Throughout this section, we set the applied magnetic field
$\mathbf h$ to zero.

\begin{figure}
  \centering
  \includegraphics[width=14cm]{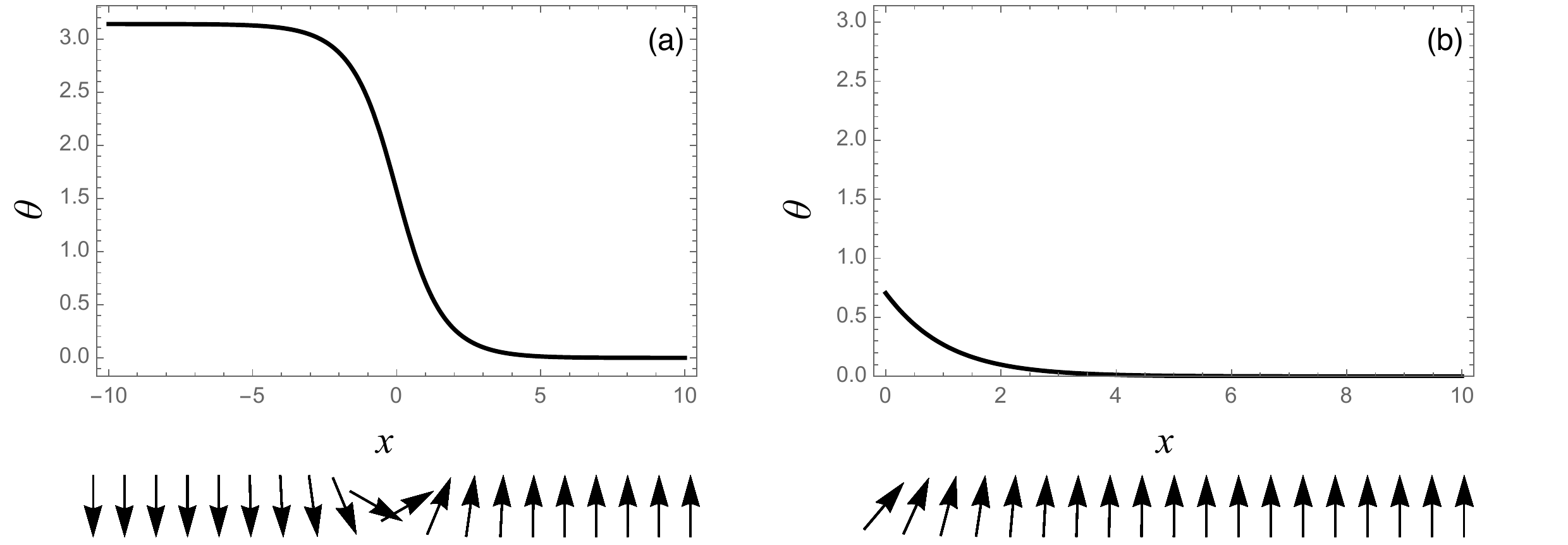}
  \caption{Two types of one-dimensional domain walls due to DMI: (a)
    interior wall; (b) edge wall. In the upper panels, $\theta$ stands
    for the angle between $\m$ and the $z$-axis. The vector $\m$
    rotates in the $xz$-plane (lower panels).}
  \label{fig:walls}
\end{figure}

\subsection{Interior wall}
\label{sec:interior-wall}

Consider first the whole space situation, in which case we may assume
that
\begin{align}
  \label{Om1d}
  \Omega = \{(x, y) \in \R^2 \ :  \ x \in \R, \ \ 0 < y < 1 \},
\end{align}
with periodic boundary conditions at $y = 0$ and $y = 1$.  We then
take $\m$ to be a one-dimensional profile, i.e., $\m = \m(x)$. Then we
may write the energy associated with $\m$ in the form
\begin{align}
  \label{E1d}
  E(\m) = \int_{-\infty}^\infty \Big\{ |\m'|^2 + (Q - 1) |\mpe|^2 + \kappa
  \left( \mpa (\hat x \cdot \mpe)' -  
  (\hat x \cdot \mpe)  \mpa' \right) \Big\} \, dx,
\end{align}
where primes denote the derivative with respect to the $x$ variable
and $\hat x$ is the unit vector in the direction of the $x$-axis. We
are interested in the global energy minimizers of the energy in
\eqref{E1d} that obey the following conditions at infinity:
\begin{align}
  \label{minf}
  \lim_{x \to \pm\infty} \mpa(x) = \pm 1, \qquad \lim_{x \to \pm \infty}
  \mpe(x) = 0. 
\end{align}

On heuristic grounds, one expects that the optimal domain wall profile
has the form of the {\em Dzyaloshinskii wall}
\cite{thiaville12}. Namely, one expects that in the domain wall the
magnetization rotates around the direction of the $y$-axis. Hence,
introducing an ansatz 
\begin{align}
  \label{manz}
  \m = (\sin \theta, 0, \cos \theta),   
\end{align}
one can rewrite the energy in \eqref{E1d} as \cite{rohart13}
\begin{align}
  \label{Eth}
  E(\m) = \int_{-\infty}^\infty \Big\{ |\theta'|^2 + (Q - 1) \sin^2
  \theta + \kappa \theta' \Big\} dx.
\end{align}

Observe, however, that a priori the energy in \eqref{Eth} is not well
defined in the natural class of $\theta \in H^1_{loc}(\R)$, since the
last term in the energy is not sign definite and does not necessarily
make sense as the Lebesgue integral on the whole real line. This fact
is closely related to the chiral nature of DMI, favoring oscillations
of the magnetization vector. A simple counterexample, in which the
first two terms of the energy in \eqref{Eth} are well defined, while
the last one is not, is given by the function
$\theta(x) = {\pi \over 2} - \mathrm{Si}(x)$, where
$\mathrm{Si}(x) = \int_0^x t^{-1} \sin t \, dt$ is the sine integral
function. It is also worth noting that if one were to define the
energy in \eqref{Eth} as the limit of the energies on large finite
domains, then its minimum value would be strictly greater than that
obtained from the integral on the whole real line due to the presence
of edge domain walls \cite{rohart13} (see also
Sec. \ref{sec:edge-domain-wall} for further details).

To fix the issue above, one needs to assume that
$\theta' \in L^1(\R)$, which introduces a bound on the total variation
of $\theta$ on $\R$. This, in turn, implies that the limit of
$\theta(x)$ as $x \to \pm \infty$ exists, and the last term in
\eqref{Eth} becomes a boundary term. Furthermore, in order for the
energy to be bounded the limits of $\theta(x)$ at infinity must be
integer multiples of $\pi$, and without loss of generality we may
assume
\begin{align}
  \label{thinf}
  \lim_{x \to -\infty} \theta(x) = \pi n, \qquad \lim_{x \to +\infty}
  \theta(x) = 0, \qquad \ n \in \mathbb Z.
\end{align}
The energy then becomes
\begin{align}
  \label{Eth2}
  E(\m) = \int_{-\infty}^\infty \Big\{ |\theta'|^2 + (Q - 1) \sin^2
  \theta \Big\} dx - \kappa \pi n, 
\end{align}
for $\theta \in H^1_{loc}(\R)$ with $ \theta' \in L^1(\R)$ and
$\theta$ obeying \eqref{thinf}, with $n \not = 0$ to exclude the
trivial case.

It is easy to see that the energy in \eqref{Eth2} is uniquely
minimized in the above class if and only if $n = 1$ and
$\kappa < \kappa_c$, where
\begin{align}
  \label{kappac}
  \kappa_c = {4 \sqrt{Q - 1} \over \pi}.
\end{align}
In this case the optimal profile is, up to translations, given by
\cite{rohart13}
\begin{align}
  \label{thmin}
  \theta(x) = 2 \arctan e^{-x \sqrt{Q-1}},
\end{align}
and the wall energy is given by
\begin{align}
  \label{Ewall}
  \sigma_{wall} = 4 \sqrt{Q - 1} - \pi \kappa > 0.
\end{align}
Indeed, minimizers of \eqref{Eth2} with $n = \pm 1$ among all
admissible $\theta$ are well known to exist due to the good coercivity
and lower semicontinuity properties of those terms (for technical
details in a related problem, see \cite{cm:non13}). The profile in
\eqref{thmin} is then the unique solution, up to translations and
sign, of the Euler-Lagrange equation associated with \eqref{Eth2}
satisfying \eqref{thinf}. At the same time, for $|n| \geq 2$ the
energy is easily seen to satisfy $E(\theta) \geq |n|
\sigma_{wall}$.
Hence, by inspection the minimizer with $n = +1$ corresponds to the
global minimizer for all $n \not = 0$, with the sign of $n$
corresponding to the wall chirality imparted by DMI.

We remark that, in contrast to the above situation, the problem
associated with \eqref{E1d} does not admit minimizers for
$\kappa > \kappa_c$, since in this case the energy is not bounded
below and favors helical structures \cite{rohart13}.

The following theorem establishes existence and uniqueness of the
minimizers of the one-dimensional domain wall energy in \eqref{E1d}
among all profiles satisfying \eqref{minf} {\em without} assuming the
ansatz in \eqref{manz}.  In view of the discussion above, an
appropriate admissible class for the energy is given by
\begin{align}
  \label{A1d}
  \mathcal A = \left\{ \m \in H^1_{loc}(\mathbb R; \mathbb S^2) \ : \ \m' 
  \in L^1(\R; \R^3) \right\}.
\end{align}
The theorem below confirms the expectation that the domain wall profile is
given by \eqref{manz} and \eqref{thmin} for all $\kappa$ below a
critical value, although the latter turns out to be slightly lower
than the expected threshold value of $\kappa = \kappa_c$ given by
\eqref{kappac}.

\begin{theorem}
  \label{t:1d}
  Let $0 < \kappa < \sqrt{Q - 1}$. Then there exists a unique, up to
  translations, minimizer $\m \in \mathcal A$ of \eqref{E1d}
  satisfying \eqref{minf}. The minimizer $\m$ has the form in
  \eqref{manz} with $\theta$ given by \eqref{thmin}, and the minimal
  energy is given by $\sigma_{wall}$ from \eqref{Ewall}.
\end{theorem}

\begin{proof}
  The proof proceeds by showing directly that the profile given by
  \eqref{manz} and \eqref{thmin} is the unique minimizer via
  establishing a sharp lower bound for the energy. Assume without loss
  of generality that $E(\m) < +\infty$. Then by dominated convergence
  theorem we have
  \begin{align}
    E(\m) = \int_{-\infty}^\infty \Big( |\m'|^2 +
    (Q - 1) |\mpe|^2 \Big) dx + \kappa \lim_{R \to \infty} \int_{-R}^R
    \left( \mpa (\hat x \cdot \mpe)' -  
    (\hat x \cdot \mpe)  \mpa' \right) dx,
  \end{align}
  and $|\mpe(x)| \to 0$ as $x \to \pm \infty$ \cite[Corollary
  8.9]{brezis}.  Using integration by parts \cite[Corollary
  8.10]{brezis}, the last integral may be rewritten as
  \begin{align}
    \int_{-R}^R \left( \mpa (\hat x \cdot \mpe)' - (\hat x \cdot \mpe)
    \mpa' \right) dx =  (\hat x \cdot \mpe(x)) \mpa(x)
    \bigg|_{-R}^R - 2 \int_{-R}^R  (\hat x \cdot \mpe) \, 
    \mpa' \, dx.
  \end{align}
  Therefore, passing to the limit we obtain that
  \begin{align}
    \label{E1d2}
    E(\m) = \int_{-\infty}^\infty \Big( |\m'|^2 + (Q - 1) |\mpe|^2 -2
    \kappa (\hat x \cdot \mpe)  \mpa' \Big) dx.
  \end{align}

  We now trivially estimate the DMI term from below to obtain
  \begin{align}
    \label{E1d3}
    E(\m) \geq \int_{-\infty}^\infty \Big( |\m'|^2 + (Q - 1) |\mpe|^2 -2
    \kappa \, |\mpe| \,  |\mpa'| \Big) dx.
  \end{align}
  Next, we use the standard trick \cite{kohn07iciam} to estimate the
  exchange energy by the term involving only $|\mpa'|$. In the
  following, we spell out the details of the argument, paying special
  attention to the optimality of the obtained estimates.  We start by
  applying the weak chain rule \cite[Proposition 9.5]{brezis} to the
  identity $|\mpe|^2 + \mpa^2 = 1$. This yields: 
  \begin{align}
    \label{chain}
    \mpa^2 |\mpa'|^2 = |\mpe \cdot \mpe'|^2 \leq |\mpe|^2 |\mpe'|^2
    \qquad \text{for a.e.} \  x \in \R.
  \end{align}
  Therefore, for a.e. $x \in \R$ such that $|\mpa| < 1$ we can write
  \begin{align}
    \label{chain2}
    {\mpa^2 |\mpa'|^2 \over 1 - \mpa^2} \leq |\mpe'|^2.
  \end{align}
  Thus
  \begin{align}
    \label{radu}
    \int_{-\infty}^\infty |\m'|^2 \, dx =  \int_{-\infty}^\infty
    \left( |\mpe'|^2 + |\mpa'|^2 \right) \, dx \geq \int_{\{|\mpa| <
    1\}}  {|\mpa'|^2 \over 1 - \mpa^2} \, dx.
  \end{align}

  Writing the lower bound for the energy in terms of $\mpa$, with the
  help of \eqref{E1d3} and \eqref{radu} we obtain
  \begin{align}
    E(\m) \geq \int_{\{|\mpa| <
    1\}} \left(  {|\mpa'|^2 \over 1 - \mpa^2} + (Q - 1) (1 - \mpa^2)
    \right) \, dx - 2 \kappa \int_{-\infty}^\infty \, \sqrt{1 -
    \mpa^2} \, |\mpa'| \, dx. 
  \end{align}
  This inequality may be rewritten in the following Modica-Mortola
  type form
  \begin{align}
    \label{MM}
    E(\m) 
    & \geq 
      2 \int_{-\infty}^\infty \left(  \sqrt{Q - 1} - \kappa \sqrt{1 -
      \mpa^2}  \, \right) |\mpa'| \, dx \notag \\ 
    & \quad + \int_{\{|\mpa| <
      1\}} \left(  {|\mpa'| \over \sqrt{1 - \mpa^2}} - \sqrt{(Q - 1)
      (1 - \mpa^2)} \right)^2 \, dx,
  \end{align}
  where we extended the domain of integration in the first term to the
  whole real line in view of the fact that by \eqref{chain} we have
  $\mpa' = 0$ whenever $|\mpa| = 1$.

  We now turn to showing that the energy is minimized by the profile
  given by \eqref{manz} with $\theta$ given by \eqref{thmin}. Indeed,
  from \eqref{MM} we have for any $R > 0$:
  \begin{align}
    \label{E1dlb}
    E(\m) & \geq 2 \int_{-R}^R \left(  \sqrt{Q - 1} - \kappa \sqrt{1 -
            \mpa^2}  \, \right) |\mpa'| \, dx \notag \\
          & \geq  2 \int_{-R}^R \left( \sqrt{Q - 1} - \kappa \sqrt{1 - \mpa^2}
            \, \right) \mpa' \, dx \notag \\
          & = \left\{ 2 \mpa(x) \sqrt{Q - 1} - \kappa
            \left(\mpa(x) \sqrt{1-\mpa^2(x)} +\arcsin(\mpa(x)) \right)
            \right\}  \bigg|_{-R}^R,
  \end{align}
  where we used the assumption that $\kappa < \sqrt{Q - 1}$ to go from
  the first to the second line. Finally, passing to the limit as
  $R \to \infty$ and using \eqref{minf}, we obtain
  \begin{align}
    E(\m) \geq \sigma_{wall},
  \end{align}
  where $\sigma_{wall}$ is defined in \eqref{Ewall}. At the same time, by
  the computation at the beginning of this section the inequality
  above is an equality when $\m$ is given by \eqref{manz} with
  $\theta$ from \eqref{thmin}.

  It remains to prove that the profile given by \eqref{manz} with
  $\theta$ from \eqref{thmin} is the unique, up to translations,
  minimizer of the energy that satisfies \eqref{minf}. Without loss of
  generality, we may assume that $\mpa(0) = 0$, in view of the
  continuity of $\mpa(x)$ and \eqref{minf}. Since the minimal value of
  the energy is attained by dropping the last term in \eqref{MM} and
  replacing $|\mpa'|$ with $\mpa'$, we have $\mpa'(x) \geq 0$ for
  a.e. $x \in \R$, and $\mpa$ satisfies
  \begin{align}
    \label{mpaopt}
    \mpa' = \sqrt{Q - 1} (1 - \mpa^2) \qquad \text{for a.e.} \ x \in
    I, 
  \end{align}
  where $I = (a, b)$ with $-\infty \leq a < 0 < b \leq \infty$. Since
  the right-hand side of \eqref{mpaopt} is continuos, $\mpa$ is the
  unique classical solution of \eqref{mpaopt} that satisfies
  $\mpa(0) = 0$, which is explicitly
  $\mpa(x) = \tanh (x \sqrt{Q - 1} \, )$. Lastly, the inequality in
  \eqref{chain} becomes equality when $\mpe'$ is parallel to $\mpe$
  and, hence, $\mpe = g \mathbf b$ for some constant vector
  $\mathbf b \in \R^2$ and a scalar function $g : \R \to [-1,1]$. In
  turn, to make an inequality in \eqref{E1d3} an equality, one needs
  to choose $\mathbf b = \hat x$ and $g \geq 0$. In view of the unit
  length constraint for $|\m|$, this translates into
  $\mpe = \hat x \, \mathrm{sech}^2 (x \, \sqrt{Q - 1})$. The obtained
  profile $\m = (\mpe, \mpa)$ is then precisely the one given by
  \eqref{manz} with $\theta$ from \eqref{thmin}.
\end{proof}

We note that the arguments in the proof of Theorem \ref{t:1d} do not
carry over to the range $\sqrt{Q - 1} < \kappa \leq \kappa_c$, since
in this range we can no longer reduce the energy by passing to the
configurations in the form given by \eqref{manz}. Nevertheless, an
inspection of the proof shows that the statement of Theorem \ref{t:1d}
remains true for all $\m = (\mpe, \mpa)$ such that $\mpa(x)$ is a
non-decreasing function of $x$. Hence, we have the following result.

\begin{theorem}
  \label{t:1dm}
  For any $\kappa > 0$, there exists a unique, up to translations,
  minimizer of \eqref{E1d} among all
  $\m = (\mpe, \mpa) \in \mathcal A$ satisfying \eqref{minf} and
  $\mpa' \geq 0$. The minimizer $\m$ has the form in \eqref{manz} with
  $\theta$ given by \eqref{thmin}, and the minimal energy is given by
  $\sigma_{wall}$ from \eqref{Ewall}.
\end{theorem}

\begin{remark}
  We point out that due to the presence of the edge domain walls (see
  the following subsection) the minimizers of the energy in \eqref{E}
  in the form of a Dzyaloshinskii wall on a strip
  $\Omega = \R \times (0, L)$ are not one-dimensional for any $L > 0$.
  Nevertheless, if one assumes periodic boundary conditions instead of
  the natural boundary conditions at the edges of the strip, an
  examination of the proof of Theorem \ref{t:1d} shows that the global
  minimizer is still given by \eqref{manz} and \eqref{thmin} in this
  case.
\end{remark}

\subsection{Edge wall}
\label{sec:edge-domain-wall}

Consider now the half-plane situation, in which case we may assume
that
\begin{align}
  \label{Om1dp}
  \Omega = \{(x, y) \in \R^2 \ :  \ x > 0, \ \ 0 < y < 1 \},
\end{align}
with periodic boundary conditions at $y = 0$ and $y = 1$.  Taking $\m$
to be a one-dimensional profile, i.e., $\m = \m(x)$, we write
\begin{align}
  \label{E12dp}
  E(\m) = \int_0^\infty \Big\{ |\m'|^2 + (Q - 1) |\mpe|^2 + \kappa
  \left( \mpa (\hat x \cdot \mpe)' -  
  (\hat x \cdot \mpe)  \mpa' \right) \Big\} \, dx,
\end{align}
where, as before, $\hat x$ is the unit vector in the direction of the
$x$-axis. Once again, in order for this energy to be bounded, we must
have $|\mpe(x)| \to 0$ as $x \to \infty$.  Hence, in view of the
symmetry
\begin{align}
  \label{mpeminmpa}
  \mpe \to -\mpe, \qquad \mpa \to -\mpa,
\end{align}
without loss of generality we may assume that
\begin{align}
  \label{minfp}
  \lim_{x \to \infty} \mpa(x) = 1.
\end{align}
Note, however, that the value of $\m(0)$ is not fixed and needs to be
determined for the optimal domain wall profile at the material
edge. Such edge domains walls were first discussed in
\cite{rohart13}. 

Since for $\kappa > \kappa_c$, where $\kappa_c$ is given by
\eqref{kappac}, the energy favors helical structures \cite{rohart13}
and, hence, is not bounded below on the semi-infinite interval as well
as on the whole line, throughout the rest of this section we assume
that $\kappa < \kappa_c$.  Assuming also the ansatz from \eqref{manz}
and arguing as in the previous subsection, for $\theta \in H^1(\R^+)$
with $\theta' \in L^1(\R^+)$ we may write the energy in \eqref{E12dp}
as
\begin{align}
  \label{Ethp}
  E(\m) = \int_0^\infty \Big\{ |\theta'|^2 + (Q - 1) \sin^2
  \theta \Big\} dx - \kappa \theta(0),   
\end{align}
which is easily seen to be minimized at fixed
$\theta(0) = \theta_0 \in (0, \pi)$ by
\begin{align}
  \label{thminp}
  \theta(x) = 2 \arctan e^{(x_0 - x) \sqrt{Q-1}}, \qquad x_0 = {\ln
  \tan \left( {\theta_0 \over 2} \right) \over \sqrt{Q-1}}. 
\end{align}
Indeed, using the Modica-Mortola trick \cite{modica87}, we rewrite the
energy in \eqref{Ethp} as
\begin{align}
  \label{MMthp}
  E(\m) = 2 \sqrt{Q - 1} \int_0^\infty |\sin \theta| \, |\theta'| \,
  dx + \int_0^\infty \Big( |\theta'| - \sqrt{Q - 1} \, |\sin
  \theta| \Big)^2 dx - \kappa \theta_0 \notag \\
  \geq  -\int_0^\infty \left( 2 \sqrt{Q - 1} \, |\sin \theta| - \kappa
  \right) \theta' dx
  = \int_0^{\theta_0} \left( 2 \sqrt{Q - 1} \, |\sin \theta| - \kappa
  \right) d \theta.
\end{align}
In particular, the inequality above becomes an equality when $\theta$
is given by \eqref{thminp}.

We now show that there exists a unique value of
$\theta_0 = \theta_0^* \in (0, \pi)$ for which the function from
\eqref{thminp} yields the absolute minimum of the energy in
\eqref{Ethp} for $\kappa < \kappa_c$. Denoting the right-hand side in
\eqref{MMthp} by $F(\theta_0)$, we observe that $F(0) = 0$,
$F'(0) < 0$, and $F(\theta_0) = F(\theta_0 - \pi) + \sigma_{wall}$,
where $\sigma_{wall} > 0$ is given by \eqref{Ewall}, for all
$\theta_0 \geq \pi$.  Therefore, for $\theta_0 \geq 0$ it is enough to
consider the values of $\theta_0 \in (0, \pi)$, for which we have
explicitly
\begin{align}
  \label{F}
  F(\theta_0) = 2 \sqrt{Q - 1} \, (1 - \cos \theta_0) - \kappa
  \theta_0.
\end{align}
A simple computation then shows that for $\theta_0 \geq 0$ the
function $F(\theta_0)$ is uniquely minimized by
\begin{align}
  \label{thpstar}
  \theta_0^* = \arcsin \left( {\kappa \over 2 \sqrt{Q - 1} } \right),
\end{align}
and the minimal value of $F(\theta_0)$ is given by
\begin{align}
  \label{Eedge}
  \sigma_{edge} = 2 \sqrt{Q - 1} \left( 1 - \sqrt{1 - {\kappa^2 \over
  4 (Q - 1)} } \, \right) - \kappa \arcsin \left( {\kappa  \over 2
  \sqrt{Q - 1} } \right) < 0. 
\end{align}
In fact, this is also an absolute lower bound for $E(\m)$ in
\eqref{Ethp}, since for $\theta_0 < 0$ the energy remains
positive. Furthermore, since $\theta_0^* \in (0, \pi)$, this minimum
value is attained by the profile in \eqref{thminp} with
$\theta_0 = \theta_0^*$. Interestingly, we find that
$\theta_0^* \in (0, \arcsin \tfrac{2}{\pi})$, spanning the range from
$0^\circ$ at $\kappa = 0$ to about $39.5^\circ$ for
$\kappa = \kappa_c$. Thus, the global minimizer of the energy in
\eqref{E12dp} among all profiles satisfying \eqref{manz} has the form
of an edge domain wall whose profile is given by \eqref{thminp}, up to
a sign, with an optimal value of $\theta$ at the edge.

We now prove, once again, that this picture remains true without the
ansatz in \eqref{manz} for a slightly smaller range of the values of
$\kappa < \kappa_c$. The appropriate admissible class for the energy
in \eqref{E12dp} is now
\begin{align}
  \label{A1dp}
  \mathcal A^+ = \left\{ \m \in H^1_{loc}(\mathbb R^+; \mathbb S^2) \
  : \ \m'  \in L^1(\R^+; \R^3) \right\}.
\end{align}

\begin{theorem}
  \label{t:edge}
  Let $0 < \kappa < \sqrt{Q - 1}$. Then there exists a unique
  minimizer $\m \in \mathcal A^+$ of \eqref{E12dp} satisfying
  \eqref{minfp}. The minimizer $\m$ has the form in \eqref{manz} with
  $\theta$ given by \eqref{thminp} and $\theta_0 = \theta_0^*$ from
  \eqref{thpstar}, and the minimal energy is given by $\sigma_{edge}$
  from \eqref{Eedge}.
\end{theorem}

\begin{proof}
  The proof proceeds exactly as in the case of Theorem \ref{t:1d},
  except that there is now an extra contribution from the boundary of
  the domain at $x = 0$. Namely, instead of \eqref{E1d2} we obtain
  \begin{align}
    \label{E1d2p}
    E(\m) = \int_0^\infty \Big( |\m'|^2 + (Q - 1) |\mpe|^2 -2
    \kappa (\hat x \cdot \mpe)  \mpa' \Big) dx - \kappa \mpa(0) (\hat
    x \cdot \mpe(0)).    
  \end{align}
  Estimating both terms coming from DMI from below as
  \begin{align}
    \label{E1d3p}
    E(\m) \geq \int_0^\infty \Big( |\m'|^2 + (Q - 1) |\mpe|^2 -2
    \kappa \, |\mpe| \,  |\mpa'| \Big) dx - \kappa \, |\mpa(0)| \,
    |\mpe(0)|,
  \end{align}
  and retracing the steps in the proof of Theorem \ref{t:1d}, we
  obtain
  \begin{align}
    \label{MMp}
    E(\m) 
    & \geq 
      2 \int_0^\infty \left(  \sqrt{Q - 1} - \kappa \sqrt{1 -
      \mpa^2}  \, \right) |\mpa'| \, dx - \kappa |\mpa(0)| \sqrt{1 -
      \mpa^2(0)} \notag \\ 
    & \quad + \int_{\{|\mpa| <
      1\}} \left(  {|\mpa'| \over \sqrt{1 - \mpa^2}} - \sqrt{(Q - 1)
      (1 - \mpa^2)} \right)^2 \, dx.
  \end{align}
  With the help of the identity $|\mpa'| = | \, |\mpa|'|$
  \cite[Theorem 6.17]{lieb-loss} and our assumption on $\kappa$, we can
  further estimate the right-hand side in \eqref{MMp} from below as
  \begin{align}
    \label{E1dlbp}
    E(\m) & \geq 2 \int_{0}^R \left(  \sqrt{Q - 1} - \kappa \sqrt{1 -
            \mpa^2}  \, \right) |\mpa'| \, dx - \kappa |\mpa(0)| \sqrt{1 -
            \mpa^2(0)}  \notag \\
          & \geq  2 \int_0^R \left( \sqrt{Q - 1} - \kappa \sqrt{1 - \mpa^2}
            \, \right) |\mpa|' \, dx - \kappa |\mpa(0)| \sqrt{1 -
            \mpa^2(0)}  \notag \\
          & = \left\{ 2 |\mpa(x)| \sqrt{Q - 1} - \kappa
            \left( |\mpa(x)| \sqrt{1-\mpa^2(x)} +\arcsin(|\mpa(x)|) \right)
            \right\}  \bigg|_{0}^R \notag \\
          & \qquad - \kappa |\mpa(0)| \sqrt{1 - \mpa^2(0)}.
  \end{align}
  Simplifying the expression above and passing to the limit, we
  arrive at
  \begin{align}
    \label{MMpp}
    E(\m) \geq 2 \sqrt{Q - 1} \, ( 1 - |\mpa(0)|) - \kappa \arccos
    |\mpa(0)|.
  \end{align}
  However, the right-hand side of \eqref{MMpp} is nothing but
  $F(\arccos |\mpa(0)|)$, where $F$ is given by \eqref{F}. Thus,
  $E(\m) \geq \sigma_{edge}$, and equality holds for the profile given
  by \eqref{manz} and \eqref{thminp}. Furthermore, as in the case of
  Theorem \ref{t:1d}, the inequality above is strict for any other
  wall profile. This concludes the proof.
\end{proof}

\begin{remark}
  \label{r:edge}
  According to Theorem \ref{t:edge}, the magnetization vector in the
  edge wall that asymptotes to $\mpa = +1$ in the sample interior
  acquires a component that points along the inner normal at the
  sample edge. At the same time, by \eqref{mpeminmpa} the
  magnetization vector in the edge wall that asymptotes to $\mpa = -1$
  in the sample interior acquires a component that points along the
  outer normal at the sample edge.
\end{remark}

\section{The problem in two dimensions}
\label{sec:two-dimens-probl0}

\begin{figure}
  \centering
  \includegraphics[width=12cm]{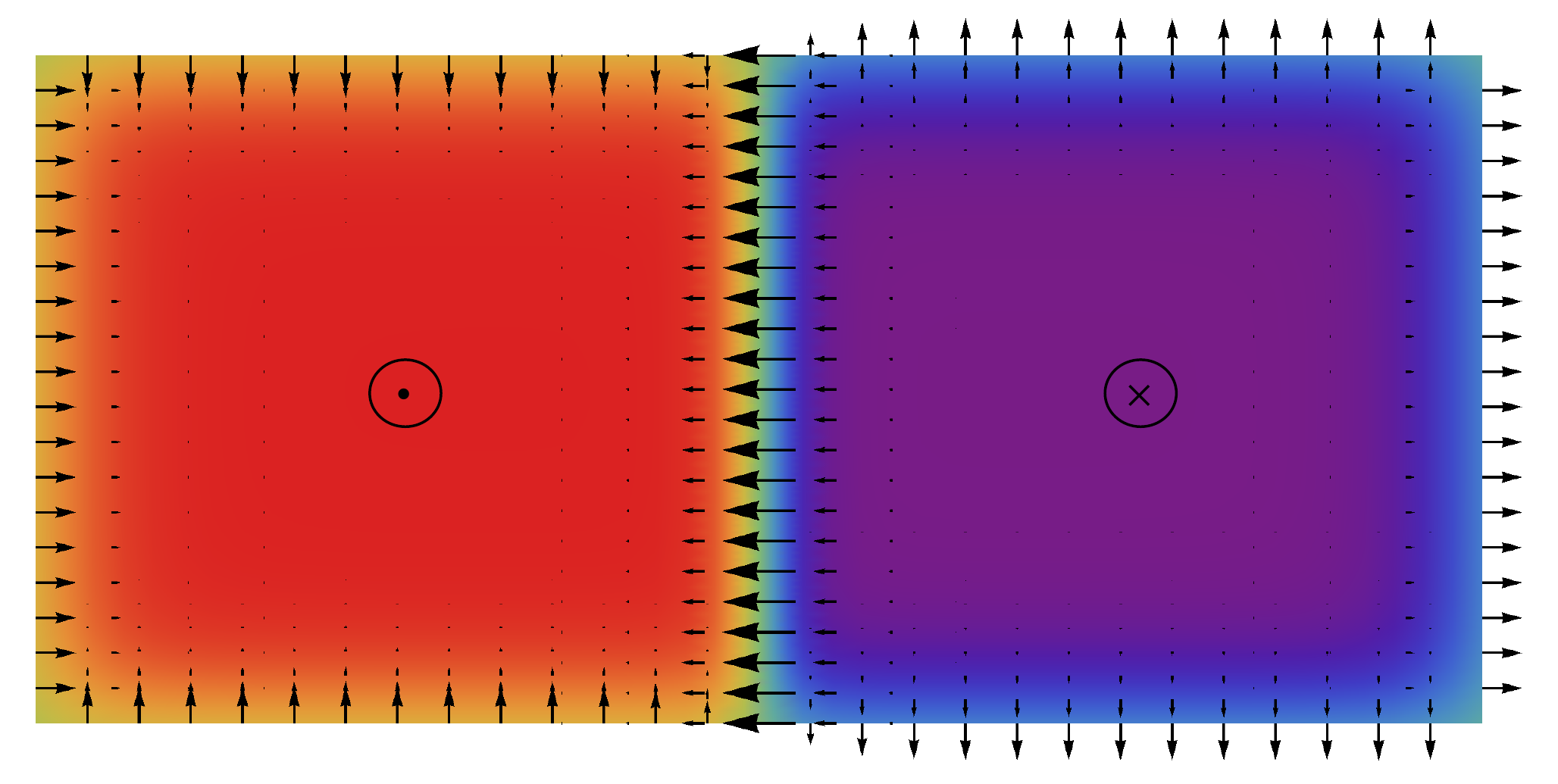}
  \caption{Schematics of a magnetization configuration containing edge
    walls and a Dzyaloshinskii wall. The arrows show the in-plane
    components of the magnetization vector, the colors correspond to
    the out-of-plane component (``red'' is up, ``violet'' is down,
    also indicated by up/down symbols).}
  \label{fig:dzyalw}
\end{figure}

We now go back to the original two-dimensional problem and consider
the regime in which the Dzyaloshinskii domain walls are present (for
an illustration, see Fig. \ref{fig:dzyalw}). The appearance of these
domain walls requires that the lateral extent of the ferromagnetic
sample be sufficiently large. Therefore, we introduce the domain
$\Omega_\eps = \eps^{-1} \Omega$, where $\eps \ll 1$, and redefine the
energy in \eqref{E} on $\Omega_\eps$:
\begin{align}
  \label{Ee}
  E(\m) = \int_{\Omega_\eps} \Big\{ |\nabla \m|^2 + (Q
  - 1) |\mpe|^2 - 2 \mathbf h_\eps \cdot \m + \kappa
  \left( \mpa \nabla \cdot   \mpe - \mpe \cdot \nabla \mpa \right)
  \Big\} \, d^2 r, 
\end{align}
where we also defined a rescaled applied field
$\mathbf h_\eps = (\hpe^\eps, \hpa^\eps) = \eps (\hpe^0, \hpa^0) =
\eps \mathbf h_0$,
chosen to have an appropriate balance between the Zeeman and the
domain wall energies (see below).  We then rescale the domain back to
$\Omega$ and the energy by a factor of $\eps$, which leads to the
following family of energies:
\begin{align}
  \label{Eeps}
  E_\eps(\m) = 
  \int_\Omega \Big\{ \eps |\nabla \m|^2 + \eps^{-1} (Q
  - 1) |\mpe|^2 - 2 \hpa^0 \mpa - 2 \hpe^0 \cdot \mpe \notag \\
  + \kappa \left(
  \mpa \nabla \cdot \mpe - \mpe \cdot \nabla \mpa \right) \Big\} \,
  d^2 r. 
\end{align}
The purpose of this section is to understand the behavior of global
energy minimizers of $E_\eps$ as $\eps \to 0$, which corresponds to
the regime of interest. Throughout the rest of this paper,
$\Omega \subset \R^2$ is assumed to be a bounded domain with boundary
of class $C^2$. This is done merely to reduce the technicalities of
the proofs and focus on the vectorial aspects of the problem involving
DMI. With slight modifications, the proof should apply to the case
when $\partial \Omega$ is a union of finitely many curve segments of
class $C^1$ (see also \cite[Remark 1.3]{modica87aihp}).

Our main tool for the analysis of the variational problem associated
with \eqref{Eeps} will be the following $\Gamma$-convergence result. 

\begin{theorem}
  \label{t:Gamma}
  Let $\mathbf h_0 = (\hpe^0, \hpa^0) \in L^\infty(\Omega; \R^3)$,
  $Q > 1$ and $0 < \kappa < \sqrt{Q - 1}$. Then, as $\eps \to 0$, we
  have $E_\eps \stackrel{\Gamma}{\longrightarrow} E_0$ with respect to
  the $L^1$ convergence, where
  \begin{align}
    \label{E0}
    E_0(\mpa) = \sigma_{edge} \mathcal H^1(\partial \Omega) +
    \sigma_{wall} \mathcal H^1(\partial^* \Omega^+) - 2
    \int_\Omega \hpa^0 \mpa \, d^2r, 
  \end{align}
  in which $\mpa \in BV(\Omega; \{-1,1\})$ and $\partial^* \Omega^+$
  is the reduced boundary of the set $\Omega^+$, where
  \begin{align}
    \label{Omp}
    \Omega^\pm = \{ x \in \Omega \ : \ \mpa(x) = \pm 1 \}.    
  \end{align}

  More precisely:
  \begin{enumerate}[i)]
  \item For any sequence of
    $\m_\eps = (\mpe^\eps, \mpa^\eps) \in H^1(\Omega; \mathbb S^2)$
    such that $\limsup_{\eps \to 0} E_\eps(\m_\eps) < +\infty$ there
    is a subsequence (not relabeled) and a function
    $\mpa^0 \in BV(\Omega; \{-1,1\})$ such that $\mpa^\eps \to \mpa^0$
    and $|\mpe^\eps| \to 0$ in $L^1(\Omega)$ as $\eps \to 0$, and
    \begin{align}
      \label{Eepsliminf}
      \liminf_{\eps \to 0} E_\eps(\m_\eps) \geq E_0(\mpa^0).
    \end{align}
  \item For any $\mpa^0 \in BV(\Omega; \{-1,1\})$ there is a sequence
    of $\m_\eps = (\mpe^\eps, \mpa^\eps) \in H^1(\Omega; \mathbb S^2)$
    such that $\mpa^\eps \to \mpa^0$ and $|\mpe^\eps| \to 0$ in
    $L^1(\Omega)$ as $\eps \to 0$, and
    \begin{align}
      \label{Eepslimsup}
      \limsup_{\eps \to 0} E_\eps(\m_\eps) \leq E_0(\mpa^0).
    \end{align}
  \end{enumerate}
\end{theorem}

\begin{proof}
  The proof follows the classical argument of Modica
  \cite{modica87aihp} adapted to the vectorial micromagnetic setting
  and taking into account the boundary contributions to the
  energy. The latter arise after integration by parts:
  \begin{align}
    \label{Eeps2}
    E_\eps(\m) = \int_\Omega \Big( \eps |\nabla \m|^2 + \eps^{-1} (Q
    - 1) |\mpe|^2 - 2 \hpa^0 \mpa - 2 \hpe^0 \cdot \mpe^0 
    - 2 \kappa \mpe \cdot \nabla \mpa
    \Big) \, d^2 r \notag \\
    + \, \kappa \int_{\partial \Omega} \wmpa (\wmpe \cdot \nu) \, d
    \mathcal H^1(r),
  \end{align}
  where $\nu$ is the outward unit normal to $\partial \Omega$ and
  $(\wmpe, \wmpa)$ is the trace of $(\mpe, \mpa)$ on
  $\partial \Omega$. The proof proceeds in three steps.

  \medskip

  \noindent \emph{Step 1: Compactness.} Given an admissible sequence
  of $\m_\eps = (\mpe^\eps, \mpa^\eps)$ satisfying
  $ E_\eps(\m_\eps) \leq C$ as $\eps \to 0$ for some $C > 0$
  independent of $\eps$, with the help of \eqref{Eeps2} and an
  elementary bound on the DMI term we can write
  \begin{align}
    \label{Eeps2comp}
    \int_\Omega \Big( \eps |\nabla \m_\eps|^2 + \eps^{-1} (Q
    - 1) |\mpe^\eps|^2 - 2 \kappa \, |\mpe^\eps| \, |\nabla \mpa^\eps| 
    \Big) \, d^2 r \qquad \quad \notag \\
    \leq C + 2 \| |\mathbf h_0 | \|_{L^\infty(\Omega)} |\Omega| +
    \kappa \mathcal H^1(\partial \Omega). 
  \end{align}
  Therefore, from \eqref{chain2} we obtain
  \begin{align}
    \label{Eeps2comp2}
    \int_{\Omega \cap \{ |\mpa^\eps| < 1 \} } \left( {\eps |\nabla
    \mpa^\eps|^2 \over 1 - |\mpa^\eps|^2} + \eps^{-1} (Q 
    - 1) (1 - |\mpa^\eps|^2) \right) \, d^2 r \qquad \quad \notag \\ 
    - 2 \kappa \int_\Omega \sqrt{1 - |\mpa^\eps|^2} \, |\nabla 
    \mpa^\eps| \, d^2 r \leq C',
  \end{align}
  for some constant $C' > 0$ independent of $\eps$. Applying the
  Modica-Mortola trick to the first line in \eqref{Eeps2comp2} and
  using the fact that by \eqref{chain} we have
  $|\nabla \mpa^\eps| = 0$ whenever $|\mpa^\eps| = 1$, we obtain
  \begin{align}
    \label{Eeps2comp3}
    2 \int_\Omega \left( \sqrt{Q - 1} - \kappa \sqrt{1 -
    |\mpa^\eps|^2} \, \right) |\nabla \mpa^\eps| \, d^2 r \leq C'. 
  \end{align}
  This is equivalent to
  $ \int_\Omega |\nabla \Phi(\mpa^\eps) | \, d^2 r \leq C'$, where
  \begin{align}
    \label{Phi2}
    \Phi(s) = 2 \int_0^s \left( \sqrt{Q - 1} - \kappa \sqrt{1 - t^2} \right)
    dt = 2  s \sqrt{Q - 1} - \kappa s \sqrt{1 - s^2} - \kappa
    \arcsin s
  \end{align}
  is a continuously differentiable, strictly increasing odd function
  of $s \in [-1,1]$.  Furthermore, by our assumption on $\kappa$ we
  have
  $0 < 2 ( \sqrt{Q - 1} - \kappa ) \leq \Phi'(s) \leq 2 \sqrt{Q - 1}$.
  Therefore, by weak chain rule \cite[Proposition 9.5]{brezis} we have
  \begin{align}
    \label{BV}
    \| \mpa^\eps \|_{W^{1,1}(\Omega)} \leq C'',     
  \end{align}
  for some $C'' > 0$ independent of $\eps$.  In turn, by compactness
  in $BV(\Omega)$ and the compact embedding of $BV(\Omega)$ into
  $L^1(\Omega)$ \cite{ambrosio}, this yields, upon extraction of a
  subsequence, that $\mpa^\eps \to \mpa^0$ in $L^1(\Omega)$ for some
  $\mpa^0 \in BV(\Omega)$.

  To prove that $|\mpa^0| = 1$ and, as a consequence, that
  $|\mpe^\eps| \to 0$ in $L^1(\Omega)$, we combine \eqref{Eeps2comp2}
  and \eqref{BV} to get
  \begin{align}
    \label{Eeps2comp4}
    \eps^{-1} (Q - 1) \int_\Omega \left( 1 - |\mpa^\eps|^2 \right) \,
    d^2 r \leq C' + 2 \kappa C''.
  \end{align}
  Therefore, the integral in the left-hand side of \eqref{Eeps2comp4}
  converges to zero as $\eps \to 0$ and, hence, $\mpa^\eps(x) \to \pm
  1$ for a.e. $x \in \Omega$. This concludes the proof of the
  compactness part of our $\Gamma$-convergence result. 

 \medskip

 \noindent \emph{Step 2: Lower bound.} We now proceed to establish
 \eqref{Eepsliminf}. By the Modica-Mortola type arguments in Step 1,
 we can estimate the energy from below as
  \begin{align}
    \label{Eeps2MM}
    E_\eps(\m_\eps) \geq \int_\Omega \left( |\nabla \Phi(\mpa^\eps)|
    - 2 \hpa^0 \mpa^\eps - 2 \hpe^0 \cdot \mpe^\eps \right) d^2 r -
    \kappa \int_{\partial\Omega} |\wmpa^\eps| \sqrt{ 1-
    |\wmpa^\eps|^2} \, d \mathcal H^1(r).  
  \end{align}
  Let $u_\eps = \Phi(\mpa^\eps)$. Then the lower bound in
  \eqref{Eeps2MM} may be rewritten as
  \begin{align}
    \label{Eeps2MM2}
    E_\eps(\m_\eps) \geq \int_\Omega \left( |\nabla u_\eps|
    - 2 \hpa^0 \mpa^\eps - 2 \hpe^0 \cdot \mpe^\eps \right) d^2 r +
    \int_{\partial\Omega} \sigma(\widetilde u_\eps) \, d \mathcal H^1(r),
  \end{align}
  where
  $\sigma(u) = -\kappa |\Phi^{-1}(u)| \sqrt{1 - |\Phi^{-1}(u)|^2}$ and
  $\widetilde u_\eps$ is the trace of $u_\eps$ on $\partial \Omega$,
  noting that $u = \Phi(s)$ defines a continuously differentiable
  one-to-one
  map from $[-1,1]$ to \\
  $I = \left[-2 \sqrt{Q - 1} + \frac12 \pi \kappa, 2 \sqrt{Q - 1} -
    \frac12 \pi \kappa \right]$.  We next define
  \begin{align}
    \label{sigmat}
    \widetilde \sigma(u) = |u| + \min_{t \in I} \left( \sigma(t) - |t|
    \right) \qquad u \in I.  
  \end{align}
  A straightforward calculation shows that we have explicitly
  \begin{align}
    \label{sigmat2}
    \widetilde \sigma(u)= |u| - \sqrt{4(Q - 1) - \kappa^2} + \kappa
    \arcsin \sqrt{1 - \frac{\kappa^2}{4(Q-1)}}.
  \end{align}
  In particular, $\widetilde \sigma(u)$ is a 1-Lipschitz function of
  $u$, and by definition $\widetilde \sigma(u) \leq \sigma(u)$.
  Therefore, by \cite[Proposition 1.2]{modica87aihp} and the fact that
  $|\mpe^\eps| \to 0$ in $L^1(\Omega)$, proved in Step 1, we have
  \begin{align}
    \label{Mlsc}
    \liminf_{\eps \to 0} E_\eps(\m_\eps) \geq \liminf_{\eps \to 0}
    \left( \int_\Omega |\nabla 
    u_\eps| \, d^2 r + \int_{\partial \Omega} \widetilde \sigma(\widetilde
    u_\eps) \, d \mathcal H^1(r) \right) -2 \int_\Omega \hpa^0
    \mpa^0 \, d^2 r \notag \\
    \geq \int_\Omega |\nabla u_0| \, d^2 r + \int_{\partial \Omega} \widetilde \sigma ( 
    \widetilde u_0 ) \, d \mathcal H^1(r) - 2 \int_\Omega \hpa^0
    \mpa^0 \, d^2 r ,
  \end{align}
  where
  $u_0 \in BV(\Omega; \{ -2 \sqrt{Q - 1} + \frac12 \pi \kappa, 2
  \sqrt{Q - 1} - \frac12 \pi \kappa \})$
  and $u_\eps \to u_0$ in $L^1(\Omega)$. In \eqref{Mlsc}, the first
  integral in the last line denotes the total variation of $u_0$, and
  the second term is understood as an integral of the trace of a BV
  function \cite{ambrosio}. Notice that by \eqref{sigmat2} we have
  $\widetilde \sigma(\widetilde u_0) = \sigma_{edge}$ and
  $|\nabla u_0| = \frac12 \sigma_{wall} |\nabla \mpa^0|$, after
  straightforward algebra. Therefore, the last inequality is
  equivalent to
  \begin{align}
    \label{Eeps2lb}
    \liminf_{\eps \to 0} E_\eps(\m_\eps) \geq {\sigma_{wall} \over 2}
    \int_\Omega | \nabla \mpa^0| \, d^2 r + \sigma_{edge} \mathcal
    H^1(\partial \Omega) - 2 \int_\Omega \hpa^0 \mpa^0 \,  d^2 r,
  \end{align}
  which coincides with \eqref{Eepsliminf} \cite{ambrosio}.

\medskip

\noindent \emph{Step 3: Upper bound.} Without loss of generality, we
may assume $\hpa=0$ and $\hpe=0$. Since we have to preserve the
constraint $|\m|=1$, we will construct an upper bound, using the angle
variables $\theta$ and $\phi$. Namely, we define
$\m = (\sin\theta \cos\phi, \sin\theta\sin\phi, \cos\theta)$ and
rewrite the energy in \eqref{Eeps} in terms of $\theta$ and $\phi$
(assumed to be sufficiently smooth) as follows:
\begin{align} 
  \label{Etp} 
  E(\m)&= \int_\Omega \Big( \eps |\nabla \theta|^2 +
         \eps \sin^2\theta
         |\nabla \phi|^2 + \eps^{-1} (Q-1) \sin^2\theta \Big) \, d^2 r
         \nonumber \\  
       & + \kappa \int_\Omega ( \sin \theta \cos \theta - \theta)  \,
         \nabla \cdot \vv(\phi)  \, d^2 r + \kappa \int_{\partial
         \Omega} \theta \, \vv(\phi) \cdot \nu \, d \mathcal H^1(r),
\end{align}
where $\vv(\phi)=(\cos\phi, \sin\phi)$, and we used integration by
parts.

Let $\Omega^\pm$ be defined as in \eqref{Omp} with $\mpa = \mpa^0$.
Without loss of generality, we assume that $\partial^* \Omega^+$ has
$C^2$ regularity, and that $\partial^* \Omega^+$ intersects
$\partial \Omega$ transversally, if at all.  We define
\begin{equation}
  \theta_*(x) =
\begin{cases}
  0& x \in \Omega^+ \\
  \pi & x \in \Omega^-
\end{cases}, \qquad \theta_b(x) =
\begin{cases}
  \theta_0^* & x \in \partial \Omega \backslash \partial \Omega^- \\
  \pi -\theta_0^* & x \in \partial \Omega \backslash \partial \Omega^+
\end{cases},
\end{equation}
where $\theta_0^*$ is defined in \eqref{thpstar}, and take a sequence
of $\theta_\eps \in C^1(\overline\Omega)$ such that
\begin{equation} 
  \label{thetaeps} 
  0 \leq \theta_\eps \leq \pi, \qquad \theta_\eps \to 
  \theta_*\hbox{ in } L^1(\Omega), \qquad \theta_\eps \to
  \theta_b \hbox{ in } L^1(\partial \Omega).
\end{equation}
Notice that we also have $\theta_\eps \to \theta_*$ in $L^q(\Omega)$
for every $q > 1$.  

Now, for a fixed $1 < p < 2$ we take two functions
$\phi_*^\pm \in W^{1,p} (\Omega^\pm)$ with values in $[0, 2 \pi)$ such
that
\begin{align}
  \vv(\tilde\phi_*^\pm(x)) = \mp \nu_{\Omega^\pm}(x) \quad \text{for
  a.e.} \quad x \in \partial \Omega^\pm,   
\end{align}
where $\nu_{\Omega^\pm}$ is the outward normal to $\Omega^\pm$ and
$\tilde \phi_*^\pm$ are the traces of $\phi_*^\pm$ on
$\partial \Omega^\pm$, respectively. Such functions exists, for
example, by \cite[Theorem 2]{marschall87}, since $\tilde \phi_*^\pm$
are $C^1$ functions of the arclength, except at a finite number of
isolated points where they have jump discontinuities, and, hence,
belong to the appropriate Besov spaces in the assumptions of
\cite{marschall87}. Next, we define $\phi_* \in W^{1,p}(\Omega)$ as
\begin{align}
  \phi_*(x) = 
  \begin{cases}
    \phi_*^-(x) & x \in \Omega^- \\
    \phi_*^+(x) & x \in \Omega^+ 
  \end{cases}
               ,
\end{align}
and observe that by construction we have
\begin{equation}
\vv(\tilde\phi_*) = 
\begin{cases}
  \nu_\Omega & \text{on} \  \partial\Omega^- \cap \partial \Omega  \\
  -\nu_\Omega & \text{on} \  \partial\Omega^+ \cap \partial \Omega \\
  \nu_* & \text{on} \ \partial^* \Omega^+
\end{cases},
\end{equation}
where $\nu$'s are the corresponding outward normals to the respective
boundaries and $\tilde \phi_*$ is the trace of $\phi_*$ on those
boundaries.  We can then construct, using a regularization and a
diagonal argument, a sequence of $\phi_\eps \in C^1(\overline\Omega)$
such that
\begin{equation}
  \phi_\eps \to \phi_* \hbox{ in } W^{1,p}(\Omega) \quad \hbox{ and }
  \quad  \eps |\nabla \phi_\eps|^2 \to 0 \hbox{ in } L^1(\Omega).
\end{equation}

It is then clear that, as $\eps \to 0$, we have
\begin{align}
  & \int_\Omega \theta_\eps \, \nabla \cdot \vv(\phi_\eps) \, d^2 r 
    \to \pi \int_{\Omega^-} \nabla \cdot \vv(\phi_*)
    \, d^2 r = \pi \mathcal H^1(\partial^* \Omega^+) + \pi \mathcal 
    H^1(\partial \Omega^- \cap \partial \Omega), \\
  & \int_\Omega \sin \theta_\eps \cos \theta_\eps \nabla \cdot
    \vv(\phi_\eps)  \, d^2 r 
    \to 0, \\
  & \int_{\partial \Omega} \theta_\eps \,  \vv(\phi_\eps) \cdot \nu\,  d
    \mathcal H^1(r) 
    \to -\theta_0^* \mathcal H^1 (\partial \Omega^+ \cap \partial
    \Omega) + (\pi -\theta_0^*) \mathcal H^1 (\partial \Omega^-
    \cap \partial \Omega).  
\end{align}
Passing to the limit as $\eps \to 0$ in the energy \eqref{Etp} and
combining the terms, we obtain
\begin{align}
  \limsup_{\eps \to 0} E(\m_\eps) 
  &= \limsup_{\eps \to 0} \int_\Omega \Big(
    \eps |\nabla \theta_\eps|^2  +
    \eps^{-1} (Q-1) \sin^2\theta_\eps \Big) \, d^2 r 
    \nonumber \\ 
  & - \pi \kappa \mathcal H^1(\partial^* \Omega^+) - \kappa \theta_0^* 
    \mathcal H^1(\partial \Omega). 
\end{align}

In order to conclude, we need to construct a sequence of
$\theta_\eps \in C^1(\overline \Omega)$ satisfying \eqref{thetaeps}
such that
\begin{align}
  \limsup_{\eps \to 0} \int_\Omega \Big( \eps |\nabla \theta_\eps|^2  +
  \eps^{-1} (Q-1) \sin^2\theta_\eps\Big) \, d^2 r \qquad \qquad \qquad
  \notag 
  \\ 
  = E_0(\mpa) + \pi \kappa \mathcal
  H^1(\partial^* \Omega^+) + \kappa \theta_0^* \mathcal H^1(\partial 
  \Omega). 
\end{align}
This construction was done in a more general setting in \cite[Lemma
2]{owen90}) and, therefore, using this result we conclude that
$\limsup_{\eps \to 0} E(\m_\eps) = E_0(\mpa)$, where \\
$\m_\eps = (\sin\theta_\eps \cos\phi_\eps, \sin\theta_\eps
\sin\phi_\eps, \cos\theta_\eps)$
and $(\theta_\eps, \phi_\eps)$ are as above.
\end{proof}

As an immediate consequence of $\Gamma$-convergence, we have the
following asymptotic characterization of minimizers of the energy
$E_\eps$ in terms of the minimizers of $E_0$.

\begin{corollary}
  Under the assumptions of Theorem \ref{t:Gamma}, let
  $\m_\eps = (\mpe^\eps, \mpa^\eps) \in H^1(\Omega; \mathbb S^2)$ be a
  sequence of minimizers of $E_\eps$. Then, after extracting a
  subsequence, we have $\mpa^\eps \to \mpa^0 $ and $|\mpe^\eps| \to 0$
  in $L^1(\Omega)$, where $\mpa^0 \in BV(\Omega; \{-1,1\})$ is a
  minimizer of $E_0$.
\end{corollary}

We note that by classical results for problems with prescribed mean
curvature (see, e.g., \cite{maggi} and references therein), the
minimizers of $E_0$ are functions, whose jump set
$\Gamma \subset \overline\Omega$ is a union of finitely many $C^{1,1}$
curve segments satisfying weakly the equation
\begin{align}
  \label{kh}
  \sigma_{wall} K(x) = 4 \hpa^0(x), \quad x \in \Gamma \cap \Omega,
  \qquad \qquad \Gamma'(x) \perp \partial \Omega, \quad x \in \Gamma 
  \cap \partial \Omega,
\end{align}
where $K$ is the curvature of $\Gamma$, positive if the set $\Omega^+$
is convex, and the prime denotes arclength derivative.  Physically,
these are interpreted as the Dzyaloshinskii domain walls separating
the domains of opposite out-of-plane magnetization under the external
applied field. We also note that the limit energy $E_0$ contains a
contribution from the edge domain walls, which, however, is
independent of the magnetization orientation near the edge and thus
only adds a constant term to the energy.

\begin{remark}
  We note that by the results of \cite{kohn89}, we can also say that
  if $\mpa^0$ is an isolated local minimizer of $E_0$, then there
  exists a sequence of local minimizers $\m_\eps$ of $E_\eps$ such
  that $\mpa^\eps \to \mpa^0$ and $\mpe^\eps \to 0$ in $L^1(\Omega)$.
\end{remark}

Before concluding this section, let us comment on some topological
issues related to the result in Theorem \ref{t:Gamma}. We note that
our upper construction in Theorem \ref{t:Gamma} uses the magnetization
configurations that have topological degree zero. This has to do with
the representation of the test configurations $\m_\eps$ adopted in the
proof in terms of the angle variables $(\theta_\eps, \phi_\eps)$,
which are assumed to be of class $C^1$ up to the boundary. Therefore,
the proof does not immediately extend to the admissible classes with
prescribed topological degree distinct from zero. This is not a
problem, however, in view of the fact that away from the domain walls
one could insert skyrmion profiles \cite{melcher14}, suitably
localized, into our test functions to prescribe a fixed topological
degree for $\eps$ sufficiently small. Our result would then not be
altered, in view of the fact that in the considered scaling the energy
of a skyrmion is a lower order perturbation to that of chiral
walls. In other words, under the considered scaling assumptions our
energy does not see magnetic skyrmions.

\section{Discussion}
\label{sec:discussion}

To summarize, we have analyzed the basic domain wall profiles in the
local version of the micromagnetic modeling framework containing DMI,
which is governed by the energy in \eqref{E}. Specifically, we
performed an analysis of the one-dimensional energy minimizing
configurations on the whole line and on half-line and showed that the
magnetization profiles expected from the physical considerations based
on specific {\em ans\"atze} are indeed the unique global energy
minimizers for $|\kappa| < \sqrt{Q - 1}$. This is slightly below
(about 30\%) the threshold value of
$|\kappa| = \kappa_c = {4 \over \pi} \sqrt{Q - 1}$, beyond which
helical structures emerge. Our methods rely on a sharp Modica-Mortola
type inequality and do not extend to the narrow range of
$\sqrt{Q - 1} \leq |\kappa| < {4 \over \pi} \sqrt{Q - 1}$. It is
natural to expect that our result persists all the way to
$|\kappa| = \kappa_c$, but to justify this statement one would need to
develop new analysis tools for the vectorial variational problem
associated with the domain walls.

Our one-dimensional analysis in section \ref{sec:probl-one-dimens}
identified two basic types of chiral domain walls: the interior and
the edge domain walls. These one-dimensional domain wall solutions are
the building blocks of the more complicated two-dimensional
magnetization configurations in ultrathin films subjected to
sufficiently small applied magnetic fields. This can be seen from the
analysis of $\Gamma$-convergence of the energy in \eqref{Eeps}
performed in section \ref{sec:two-dimens-probl0}. Either global or
local energy minimizers for $\eps \ll 1$ may then be approximated by
those of the energy in \eqref{E0}, which determines the geometry of
the magnetic domains in the sample. Our findings indicate that in the
considered limit the magnetization configurations solve the prescribed
mean curvature problem in \eqref{kh}, again, for
$|\kappa| < \sqrt{ Q - 1 }$.  We note that our variational setting
could similarly be used to study the gradient flow dynamics governed
by \eqref{Eeps} (for a related study, see \cite{owen90}). Other
physical effects, however, need to be incorporated to account for some
unusual properties of chiral domain walls such as their tilt in
sufficiently strong external fields \cite{boulle13,mst16}.

Finally, we would like to comment on the assumptions that lead to the
model in \eqref{Eeps}, and on its possible generalizations. As was
already mentioned, this energy functional is local, with the effect of
the stray field surviving in the renormalized magnetocrystalline
anisotropy term only. This is justified in the limit of arbitrarily
thin ferromagnetic films \cite{gioia97}. In practice, this
contribution is only the leading order term in the expansion of the
energy in the film thickness for films whose thickness is less than
the exchange length $\ell_{ex}$ of the material. Going to higher
order, two types of contributions appear. The first is the one coming
from the sample boundary. In the limit of the dimensionless film
thickness $\delta = d / \ell_{ex}$ going to zero, this contribution
becomes local and adds an extra penalty term for the in-plane
component of the magnetization at the edge \cite{kohn05arma}:
\begin{align}
  \label{Ebdry}
  E_\eps^{edge} (\m) = {\delta |\ln \delta| \over 2 \pi}
  \int_{\partial \Omega} (\nu \cdot \mpe)^2 \, d \mathcal H^1(r),
\end{align}
where $\nu$ is the outward unit normal to $\partial\Omega$.  Here we
took into account that in a perpendicular material the magnetic
``charge'' at the sample boundary would be smeared on the scale of
$\ell_{ex}$.  In the interior, the leading order contribution from the
stray field energy beyond the shape anisotropy can be shown to be
\cite{kmn16}:
\begin{align}
  \label{Ebulk}
  E_\eps^{bulk} (\m) = - {\delta \over 8 \pi} \int_\Omega \int_\Omega
  {(\mpa(\mathbf r) - \mpa(\mathbf r'))^2 \over |\mathbf r - \mathbf
  r'|^3} \, d^2 r \, d^2 r' +  {\delta \over 4 \pi}  \int_\Omega
  \int_\Omega {\nabla \cdot \mpe(\mathbf r) \, \nabla \cdot \mpe
  (\mathbf r') \over | \mathbf r - \mathbf r'|} \, d^2 r \, d^2 r'.  
\end{align}
Furthermore, for $\delta = \lambda |\ln \eps|^{-1}$ it was shown in
the case $\kappa = 0$ and periodic boundary conditions in the plane
that as $\eps \to 0$ the effect of the stray field energy is to
renormalize the one-dimensional wall energy to a lower value, as long
as $\lambda < \lambda_c = 2 \pi \sqrt{Q - 1}$ \cite{kmn16}. It is
natural to expect from the results of \cite{kmn16} that, as
$\eps \to 0$, the wall energy for $\kappa > 0$ will become
\begin{align}
  \sigma_{wall} = 4 \sqrt{Q - 1} - \pi \kappa - {2 \lambda \over \pi}.
\end{align}
Similarly, one would expect that in this regime the edge wall energy
$\sigma_{edge}$ would also be renormalized to minimize the sum of the
exchange, anisotropy, DMI energies (all contained in \eqref{Eeps}) and
the stray field energy contributions from \eqref{Ebdry} and
\eqref{Ebulk}. This study is currently underway. At the same time, for
$\lambda > \lambda_c$ one expects spontaneous onset of milti-domain
magnetization patterns and qualitatively new system behavior (for a
recent experimental illustration, see \cite{woo16}).

\paragraph*{Acknowledgements.}

The work of CBM was supported, in part, by NSF via grants DMS-1313687
and DMS-1614948. VS would like to acknowledge support from EPSRC grant
EP/K02390X/1 and Leverhulme grant RPG-2014-226.

\bibliography{../../nonlin,../../mura,../../stat}

\begin{thebibliography}{10}

\bibitem{allwood05}
D.~A. Allwood, G.~Xiong, C.~C. Faulkner, D.~Atkinson, D.~Petit, and R.~P.
  Cowburn.
\newblock Magnetic domain-wall logic.
\newblock {\em Science}, 309:1688--1692, 2005.

\bibitem{ambrosio}
L.~Ambrosio, N.~Fusco, and D.~Pallara.
\newblock {\em Functions of bounded variation and free discontinuity problems}.
\newblock Oxford Mathematical Monographs. The Clarendon Press, New York, 2000.

\bibitem{bader10}
S.~D. Bader and S.~S.~P. Parkin.
\newblock Spintronics.
\newblock {\em Ann. Rev. Cond. Mat. Phys.}, 1:71--88, 2010.

\bibitem{bode07}
M.~Bode, M.~Heide, K.~von Bergmann, P.~Ferriani, S.~Heinze, G.~Bihlmayer,
  A.~Kubetzka, O.~Pietzsch, S.~Blugel, and R.~Wiesendanger.
\newblock Chiral magnetic order at surfaces driven by inversion asymmetry.
\newblock {\em Nature}, 447:190--193, 2007.

\bibitem{bogdanov94}
A.~Bogdanov and A.~Hubert.
\newblock Thermodynamically stable magnetic vortex states in magnetic crystals.
\newblock {\em J. Magn. Magn. Mater.}, 138:255--269, 1994.

\bibitem{bogdanov89}
A.~N. Bogdanov and D.~A. Yablonskii.
\newblock Thermodynamically stable ``vortices'' in magnetically ordered
  crystals. {The} mixed state of magnets.
\newblock {\em Sov. Phys. -- JETP}, 68:101--103, 1989.

\bibitem{boulle13}
O.~Boulle, S.~Rohart, L.~D. Buda-Prejbeanu, E.~Ju\'e, I.~M. Miron, S.~Pizzini,
  J.~Vogel, G.~Gaudin, and A.~Thiaville.
\newblock Domain wall tilting in the presence of the {Dzyaloshinskii-Moriya}
  interaction in out-of-plane magnetized magnetic nanotracks.
\newblock {\em Phys. Rev. Lett.}, 111:217203, 2013.

\bibitem{brataas12}
A.~Brataas, A.~D. Kent, and H.~Ohno.
\newblock Current-induced torques in magnetic materials.
\newblock {\em Nature Mat.}, 11:372--381, 2012.

\bibitem{braun12}
H.-B. Braun.
\newblock Topological effects in nanomagnetism: from superparamagnetism to
  chiral quantum solitons.
\newblock {\em Adv. Physics}, 61:1--116, 2012.

\bibitem{brezis}
H.~Brezis.
\newblock {\em Functional Analysis, Sobolev Spaces and Partial Differential
  Equations}.
\newblock Springer, 2011.

\bibitem{chen13}
G.~Chen, T.~Ma, A.~T. N'Diaye, H.~Kwon, C.~Won, Y.~Wu, and A.~K. Schmid.
\newblock Tailoring the chirality of magnetic domain walls by interface
  engineering.
\newblock {\em Nature Commun.}, 4:2671 pp. 1--6, 2013.

\bibitem{cm:non13}
M.~Chermisi and C.~B. Muratov.
\newblock One-dimensional {N\'eel} walls under applied external fields.
\newblock {\em Nonlinearity}, 26:2935--2950, 2013.

\bibitem{crepieux98}
A.~Cr{\'e}pieux and C.~Lacroix.
\newblock {Dzyaloshinsky--Moriya} interactions induced by symmetry breaking at
  a surface.
\newblock {\em J. Magn. Magn. Mater.}, 182:341--349, 1998.

\bibitem{desimone06r}
A.~DeSimone, R.~V. Kohn, S.~M\"uller, and F.~Otto.
\newblock Recent analytical developments in micromagnetics.
\newblock In G.~Bertotti and I.~D. Mayergoyz, editors, {\em The Science of
  Hysteresis}, volume~2 of {\em Physical Modelling, Micromagnetics, and
  Magnetization Dynamics}, pages 269--381. Academic Press, Oxford, 2006.

\bibitem{dzyaloshinskii58}
I.~Dzyaloshinskii.
\newblock A thermodynamic theory of ``weak'' ferromagnetism of
  antiferromagnetics.
\newblock {\em J. Phys. Chem. Solids}, 4:241--255, 1958.

\bibitem{emori13}
S.~Emori, U.~Bauer, S.-M. Ahn, E.~Martinez, and G.~S.~D. Beach.
\newblock Current-driven dynamics of chiral ferromagnetic domain walls.
\newblock {\em Nature Mat.}, 12:611--616, 2013.

\bibitem{fert90}
A.~Fert.
\newblock Magnetic and transport-properties of metallic multilayers.
\newblock {\em Mater. Sci. Forum}, 59:439--480, 1990.

\bibitem{fert13}
A.~Fert, V.~Cros, and J.~Sampaio.
\newblock Skyrmions on the track.
\newblock {\em Nature Nanotechnol.}, 8:152--156, 2013.

\bibitem{fert80}
A.~Fert and P.~M. Levy.
\newblock Role of anisotropic exchange interactions in determining the
  properties of spin-glasses.
\newblock {\em Phys. Rev. Lett.}, 44:1538--1541, 1980.

\bibitem{fonseca89}
I.~Fonseca and L.~Tartar.
\newblock The gradient theory of phase transitions for systems with two
  potential wells.
\newblock {\em Proc. Roy. Soc. Edinburgh Sect. A}, 111:89--102, 1989.

\bibitem{gioia97}
G.~Gioia and R.~D. James.
\newblock Micromagnetics of very thin films.
\newblock {\em Proc. R. Soc. Lond. Ser. A}, 453:213--223, 1997.

\bibitem{goussev13}
A.~Goussev, R.~G. Lund, J.~M. Robbins, V.~Slastikov, and C.~Sonnenberg.
\newblock Domain wall motion in magnetic nanowires: an asymptotic approach.
\newblock {\em Proc. R. Soc. Lond. Ser. A Math. Phys. Eng. Sci.}, 469:20130308,
  2013.

\bibitem{heinrich93}
B.~Heinrich and J.~F. Cochran.
\newblock Ultrathin metallic magnetic films: magnetic anisotropies and exchange
  interactions.
\newblock {\em Adv. Phys.}, 42:523--639, 1993.

\bibitem{heinze11}
S.~Heinze, K.~von Bergmann, M.~Menzel, J.~Brede, A.~Kubetzka, R.~Wiesendanger,
  G.~Bihlmayer, and S.~Blugel.
\newblock Spontaneous atomic-scale magnetic skyrmion lattice in two dimensions.
\newblock {\em Nature Phys.}, 7:713--718, 2011.

\bibitem{hrabec14}
A.~Hrabec, N.~A. Porter, A.~Wells, M.~J. Benitez, G.~Burnell, S.~McVitie,
  D.~McGrouther, T.~A. Moore, and C.~H. Marrows.
\newblock Measuring and tailoring the {Dzyaloshinskii-Moriya} interaction in
  perpendicularly magnetized thin films.
\newblock {\em Phys. Rev. B}, 90:020402, 2014.

\bibitem{hubert}
A.~Hubert and R.~Sch\"afer.
\newblock {\em Magnetic Domains}.
\newblock Springer, Berlin, 1998.

\bibitem{ikeda10}
S.~Ikeda, K.~Miura, H.~Yamamoto, K.~Mizunuma, H.~D. Gan, M.~Endo, S.~Kanai,
  J.~Hayakawa, F.~Matsukura, and H.~Ohno.
\newblock A perpendicular-anisotropy {CoFeB--MgO} magnetic tunnel junction.
\newblock {\em Nature Mat.}, 9:721--724, 2010.

\bibitem{kmn16}
H.~Kn\"upfer, C.~B. Muratov, and F.~Nolte.
\newblock Magnetic domains in thin ferromagnetic films with strong
  perpendicular anisotropy.
\newblock Preprint, 2016.

\bibitem{kohn07iciam}
R.~V. Kohn.
\newblock Energy-driven pattern formation.
\newblock In {\em International {C}ongress of {M}athematicians. {V}ol. {I}},
  pages 359--383. Eur. Math. Soc., Z\"urich, 2007.

\bibitem{kohn05arma}
R.~V. Kohn and V.~V. Slastikov.
\newblock Another thin-film limit of micromagnetics.
\newblock {\em Arch. Ration. Mech. Anal.}, 178:227--245, 2005.

\bibitem{kohn89}
R.~V. Kohn and P.~Sternberg.
\newblock Local minimisers and singular perturbations.
\newblock {\em Proc. Roy. Soc. Edinburgh Sect. A}, 111:69--84, 1989.

\bibitem{lieb-loss}
E.~H. Lieb and M.~Loss.
\newblock {\em Analysis}.
\newblock American Mathematical Society, Providence, RI, 2010.

\bibitem{maggi}
F.~Maggi.
\newblock {\em Sets of Finite Perimeter and Geometric Variational Problems}.
\newblock Cambridge Studies in Advanced Mathematics, 135. Cambridge University
  Press, Cambridge, 2012.

\bibitem{marschall87}
J.~Marschall.
\newblock The trace of {Sobolev-Slobodeckij} spaces on {Lipschitz} domains.
\newblock {\em Manuscripta Math.}, 58:47--65, 1987.

\bibitem{matsukura15}
F.~Matsukura, Y.~Tokura, and H.~Ohno.
\newblock Control of magnetism by electric fields.
\newblock {\em Nature Nanotechnol.}, 10:209--220, 2015.

\bibitem{melcher14}
C.~Melcher.
\newblock Chiral skyrmions in the plane.
\newblock {\em Proc. R. Soc. Lond. Ser. A}, 470:0394 pp. 1--17, 2014.

\bibitem{modica87}
L.~Modica.
\newblock The gradient theory of phase transitions and the minimal interface
  criterion.
\newblock {\em Arch. Rational Mech. Anal.}, 98:123--142, 1987.

\bibitem{modica87aihp}
L.~Modica.
\newblock Gradient theory of phase transitions with boundary contact energy.
\newblock {\em {Ann. Inst. Henri Poincar\'e. Anal. Non Lin\'eaire}},
  4:487--512, 1987.

\bibitem{moriya60}
T.~Moriya.
\newblock Anisotropic superexchange interaction and weak ferromagnetism.
\newblock {\em Phys. Rev.}, 120:91--98, 1960.

\bibitem{mst16}
C.~B. Muratov, V.~V. Slastikov, and O.~A. Tretiakov.
\newblock Theory of tilted {Dzyaloshinskii} walls in the presence of in-plane
  magnetic fields.
\newblock (In preparation), 2016.

\bibitem{nagaosa13}
N.~Nagaosa and Y.~Tokura.
\newblock Topological properties and dynamics of magnetic skyrmions.
\newblock {\em Nature Nanotechnol.}, 8:899--911, 2013.

\bibitem{owen90}
N.~C. Owen, J.~Rubinstein, and P.~Sternberg.
\newblock Minimizers and gradient flows for singularly perturbed bi-stable
  potentials with a {Dirichlet} condition.
\newblock {\em Proc. R. Soc. Lond. Ser. A}, 429:505--532, 1990.

\bibitem{parkin08}
S.~S.~P. Parkin, M.~Hayashi, and L.~Thomas.
\newblock Magnetic domain-wall racetrack memory.
\newblock {\em Science}, 320:190--194, 2008.

\bibitem{prinz98}
G.~A. Prinz.
\newblock Magnetoelectronics.
\newblock {\em Science}, 282:1660--1663, 1998.

\bibitem{rohart13}
S.~Rohart and A.~Thiaville.
\newblock Skyrmion confinement in ultrathin film nanostructures in the presence
  of {Dzyaloshinskii-Moriya} interaction.
\newblock {\em Phys. Rev. B}, 88:184422, 2013.

\bibitem{romming13}
N.~Romming, C.~Hanneken, M.~Menzel, J.~E. Bickel, B.~Wolter, K.~von Bergmann,
  A.~Kubetzka, and R.~Wiesendanger.
\newblock Writing and deleting single magnetic skyrmions.
\newblock {\em Science}, 341:636--639, 2013.

\bibitem{sampaio13}
J.~Sampaio, V.~Cros, S.~Rohart, A.~Thiaville, and A.~Fert.
\newblock Nucleation, stability and current-induced motion of isolated magnetic
  skyrmions in nanostructures.
\newblock {\em Nature Nanotechnol.}, 8:839--844, 2013.

\bibitem{stepanova}
M.~Stepanova and S.~Dew, editors.
\newblock {\em Nanofabrication: Techniques and Principles}.
\newblock Springer-Verlag, Wien, 2012.

\bibitem{sternberg91}
P.~Sternberg.
\newblock Vector-valued local minimizers of nonconvex variational problems.
\newblock {\em Rocky Mountain J. Math.}, 21:799--807, 1991.

\bibitem{thiaville12}
A.~Thiaville, S.~Rohart, E.~Ju{\'e}, V.~Cros, and A.~Fert.
\newblock Dynamics of {Dzyaloshinskii} domain walls in ultrathin magnetic
  films.
\newblock {\em Europhys. Lett.}, 100:57002, 2012.

\bibitem{vonbergmann14}
K.~von Bergmann, A.~Kubetzka, O.~Pietzsch, and R.~Wiesendanger.
\newblock Interface-induced chiral domain walls, spin spirals and skyrmions
  revealed by spin-polarized scanning tunneling microscopy.
\newblock {\em J. Phys. -- Condensed Matter}, 26:394002, 2014.

\bibitem{woo16}
S.~Woo, K.~Litzius, B.~Kruger, M.-Y. Im, L.~Caretta, K.~Richter, M.~Mann,
  A.~Krone, R.~M. Reeve, M.~Weigand, P.~Agrawal, I.~Lemesh, M.-A. Mawass,
  P.~Fischer, M.~Klaui, and G.~S.~D. Beach.
\newblock Observation of room-temperature magnetic skyrmions and their
  current-driven dynamics in ultrathin metallic ferromagnets.
\newblock {\em Nature Mat.}, 15:501--506, 2016.

\bibitem{zhang15sr}
X.~Zhang, M.~Ezawa, and Y.~Zhou.
\newblock Magnetic skyrmion logic gates: conversion, duplication and merging of
  skyrmions.
\newblock {\em Scientific Reports}, 5:9400, 2015.

\bibitem{zutic04}
I.~Zutic, J.~Fabian, and S.~Das~Sarma.
\newblock Spintronics: Fundamentals and applications.
\newblock {\em Rev. Mod. Phys.}, 76:323--410, 2004.

\end{thebibliography}
\bibliographystyle{plain}

\end{document}